\documentclass[sigconf, 9pt, nonacm]{acmart}

\renewcommand\footnotetextcopyrightpermission[1]{}

\newtheorem{example}{Example}
\usepackage{makecell}
\usepackage{booktabs,tabularx,array}
\newcolumntype{Y}{>{\raggedright\arraybackslash}X}

\newcolumntype{N}[1]{>{\raggedright\arraybackslash}p{#1}}

\usepackage{algorithm}
\usepackage{algpseudocode}
\usepackage{amsmath,amssymb,amsfonts}
\usepackage{graphicx}
\usepackage{textcomp}
\usepackage{xcolor}
\usepackage{adjustbox}
\usepackage{multirow}
\usepackage{pifont}
\usepackage{makecell} 
\usepackage[most]{tcolorbox}
\usepackage{enumitem}
\usepackage{xspace}

\usepackage{pifont}
\newcommand{\xmark}{\ding{55}}

\usepackage{booktabs}
\usepackage{siunitx}
\usepackage[table]{xcolor} 
\usepackage{tikz}          
\usepackage{tcolorbox}
\tcbuselibrary{listings}
\newtcolorbox{promptbox}[2][]{%
  colback=blue!5,        
  colframe=blue!60,      
  colbacktitle=blue!20,  
  coltitle=blue!50!black,        
  fonttitle=\bfseries,   
  fontupper=\ttfamily\small,   
  boxrule=0.6pt,         
  arc=2mm,               
  top=1mm, bottom=1mm,
  left=1mm, right=1mm,
  title=#2,#1
}

\begin{document}

\title{ Security Analysis for SCONE Logic Locking }

\author{Akashdeep Saha$^\dagger$, Mohammed Nabeel$^\ddagger$, Johann Knechtel$^\dagger$, Michail Maniatakos$^\dagger$, Ozgur Sinanoglu$^\dagger$ }
\affiliation{%
  \institution{$^\dagger$New York University Abu Dhabi, Abu Dhabi}
  \country{United Arab Emirates} \\
  \institution{$^\ddagger$New York University Tandon School of Engineering, New York}
  \country{United States}
}

\email{{as19360, mtn2, johann, michail.maniatakos, ozgursin}@nyu.edu}

\begin{abstract}
SCONE [DAC'25] expands a logic locking interface with additional encoded
inputs derived from the original primary inputs, and admits two
realizations: a \textit{with-ES} variant, where the critical encoding stage is
implemented in hardware, and a \textit{without-ES} variant, where the locked
design directly exposes an encoded interface of width $n+m$. We show
that both realizations are vulnerable, but for different reasons.
For the without-ES variant, we prove that, when the added encoded
inputs are deterministic linear functions of the original inputs, the
valid encoded-input space remains $n$-dimensional despite the nominal
expansion to $n+m$ inputs. Hence, the widened interface 
does not yield $m$ additional or independent brute-force dimensions.
For the with-ES variant, we present a polynomial-time white-box attack
that exactly recovers the added-input count and the implemented
linear encoding relation from the locked netlist, achieving 100\%
recovery over all evaluated instances. We also develop a black-box
procedure that certifies the same dimensionality collapse from valid
encoded-input samples without reconstructing the hidden encoder.
Experiments on ISCAS-85 and ITC-99 benchmarks validate both results,
and we further demonstrate exact white-box recovery on an ARM
Cortex-M0 RTL benchmark. Finally, we propose a lightweight
non-linear mitigation and show that it does not exhibit the
vulnerabilities identified in this paper under all representative
attack sets considered in SCONE.
\end{abstract}

\maketitle

\section{Introduction}
Logic locking (LL) is a widely studied hardware security technique for protecting integrated circuits against a variety of malicious activities in an untrusted supply-chain setting. Over the last decade, LL research has shifted toward SAT-resilient constructions, motivated by the vulnerability of conventional schemes to oracle-guided deobfuscation. Among these, SFLL-style schemes~\cite{yasin2017provably,sengupta2020truly, saha2025midas} have
been especially influential because they combine functional
corruption with restore logic to harden the design against classical
SAT~\cite{subramanyan2015evaluating} and removal attacks~\cite{yasin2017removal}.

At the same time, a central lesson from the LL literature is that resistance to SAT-based attacks does not imply resistance to implementation-level structural and functional analysis~\cite{han2021does,sirone2020functional,limaye2022valkyrie,alrahis2021gnnunlock}. A sequence of white-box, functional, and learning-based attacks has shown that protection logic often leaves exploitable traces in the synthesized netlist. This gap between construction-level rationale and implementation-level behavior is particularly important for
recent provably secure logic locking schemes.

SCONE~\cite{han2025scone} is a recent SFLL-style logic locking scheme that augments a circuit with additional encoded inputs derived from the original primary inputs. It admits two realizations. In the \emph{with-ES}
realization, the encoding stage (ES) is implemented in hardware and remains part of the locked netlist. In the \emph{without-ES} realization, the encoding is applied outside the locked design, and the netlist directly exposes an encoded interface of width $n+m$. These two realizations lead to different security questions and, as we show in this paper, to different vulnerabilities.

For the \emph{without-ES} realization, SCONE's security rationale is
tied to the nominal expansion of the visible interface from $n$ to
$n+m$ inputs. We show that this argument does not hold when the added
encoded inputs are deterministic linear functions of the original
inputs: the valid encoded-input set remains confined to an
$n$-dimensional subspace of the nominal $(n+m)$-dimensional space
interface. Thus, the visible interface expansion does not provide
$m$ additional independent brute-force dimensions. For the
\emph{with-ES} realization, the encoder itself becomes a white-box target. Because the added encoded inputs are implemented as deterministic linear functions of the original primary inputs, the planted encoding relation can be recovered exactly from the locked netlist in polynomial time. Hence, the hardware-implemented ES introduces a constructive structural vulnerability rather than a meaningful security gain.

Beyond the vulnerability analysis, we also propose a lightweight
mitigation that preserves the overall SCONE design pattern while
removing the specific linear structure exploited by our attacks. The key idea is to replace the linear ES with a non-linear
one, thereby preventing exact recovery as a single matrix over
$\mathbb{F}_2$. We evaluate this modified construction against the
same representative structural and I/O-based attacks considered in
the SCONE setting~\cite{han2025scone}, including the vulnerability
identified in this paper, and do not observe these attacks on the
resulting instances.

We validate the proposed analyses on ISCAS-85 and ITC-99 benchmarks, where the white-box attack achieves exact recovery across all evaluated with-ES instances, and the black-box experiments
empirically confirm dimensionality collapse in the without-ES
setting. We further evaluate white-box recovery on an ARM Cortex-M0
RTL benchmark, demonstrating exact recovery of the added encoded
inputs on a substantially larger real-world design.

The main contributions of this work are as follows:
\begin{itemize}
\item We present a two-sided security analysis of SCONE and its
two variants: an exact polynomial-time white-box recovery for with-ES realization and a black-box dimensional refutation of entropy argument for the without-ES variant.

\item We validate both analyses experimentally on standard benchmarks and on an ARM Cortex-M0 RTL benchmark.

\item We propose a non-linear mitigation and show, head-to-head
against the representative attack set considered in the SCONE
context, that the modified construction does not exhibit the
vulnerability identified in this paper.

\end{itemize}

The rest of this paper is organized as follows. Section~\ref{sec:bg}
reviews SCONE, representative logic locking attacks, and defenses. Section~\ref{subsec:entropy_vuln} presents the vulnerability analysis for both the with-ES and the without-ES settings. Section~\ref{sec:Exp} reports the experimental validation and empirical estimates. Section~\ref{sec:Dis} discusses a brief mitigation direction, and we conclude in Section~\ref{sec:Conclu}.

\section{Background and Related Work}
\label{sec:bg}
Following the SAT attack~\cite{subramanyan2015evaluating}, research in LL shifted from
conventional schemes toward SAT-resilient families, such as
\cite{yasin2016sarlock,xie2018anti,Yasin2017SFLL,Shakya2020CASLock,saha2020lopher}, and other provably secure LL techniques. Subsequent work, however, showed that resistance to SAT-based attacks does not by itself imply resistance to implementation-level structural and functional analysis~\cite{sirone2020functional,roy2026nilopher,saha2022dip}. Recent studies have increasingly leveraged large language models to support benchmarking and automation in logic locking workflows~\cite{saha2025gllamor,saha2025lockforge,basu2026poster}.

SCONE is a LL scheme~\cite{han2025scone} that augments an
$n$-input circuit with $m$ additional encoded inputs, yielding an
encoded interface of size $n+m$. In their proposed linear construction, 
each added encoded input is an XOR-based function of the
original primary inputs over $\mathbb{F}_2$. Figure~\ref{fig:scone_diag}
illustrates the two SCONE realizations proposed in~\cite{han2025scone}. In Figure~\ref{fig:scone_diag}(a), the \emph{with-ES} realization, the ES is implemented in hardware, and the locked design retains the original $n$-input interface. The ES maps the original inputs to an encoded interface, which is then processed by the locked datapath, including $FSC_{\textit{encoded}}$ and the restore unit, whose outputs are combined to produce the locked output. In Figure~\ref{fig:scone_diag}(b), the \emph{without-ES} realization, the encoding is applied outside the
locked design. Hence, the netlist directly exposes the encoded interface of width $n+m$, while the ES itself is not present in hardware. Consequently, the apparent input space expands from $2^n$ to $2^{n+m}$ only in the without-ES realization.

A major lesson from the literature is that SAT resilience does not
imply resilience to structural analysis. FALL~\cite{sirone2020functional}
showed that structural and functional artifacts in locked circuits can
be exploited to localize protection-relevant nodes and recover
locking information directly from the locked netlist, often without oracle access. Valkyrie~\cite{limaye2022valkyrie} later systematized this perspective for PSLL by diagnosing implementation-level vulnerabilities and isolating the logic associated with the locking mechanism. Likewise, attacks on CAS-Lock and its variants~\cite{sengupta2021breaking,saha2022dip} showed that hardware-level structural traces can invalidate abstract security claims, while KRATT~\cite{aksoy2024kratt} further reinforced the practicality of removal-and-structural analysis against SAT-resilient locking. Likewise, learning-based structural exploitation has shown that locking-specific logic can also be identified directly from netlist structure using data-driven models, enabling oracle-less analysis and removal of protection logic~\cite{chakraborty2018sail,alrahis2021gnnunlock,kamali2022LLadvances}.

These observations are directly relevant to the assumptions used in this paper. Structural attacks on LL do not treat the locked circuit monolithically; instead, they first localize the protection-relevant region and then analyze the resulting reduced logic. Our white-box analysis of SCONE follows the same principle: we first identify candidate ES-output wires from the SCONE-locked netlist and then analyze the corresponding reduced structure. Thus, the attack model adopted here is consistent with established structural-analysis practice in LL~\cite{limaye2022valkyrie,sirone2020functional}. Next, we elaborate on the proposed exploitations.

\begin{figure}[tb]
  \centering
  \includegraphics[width=0.65\columnwidth]{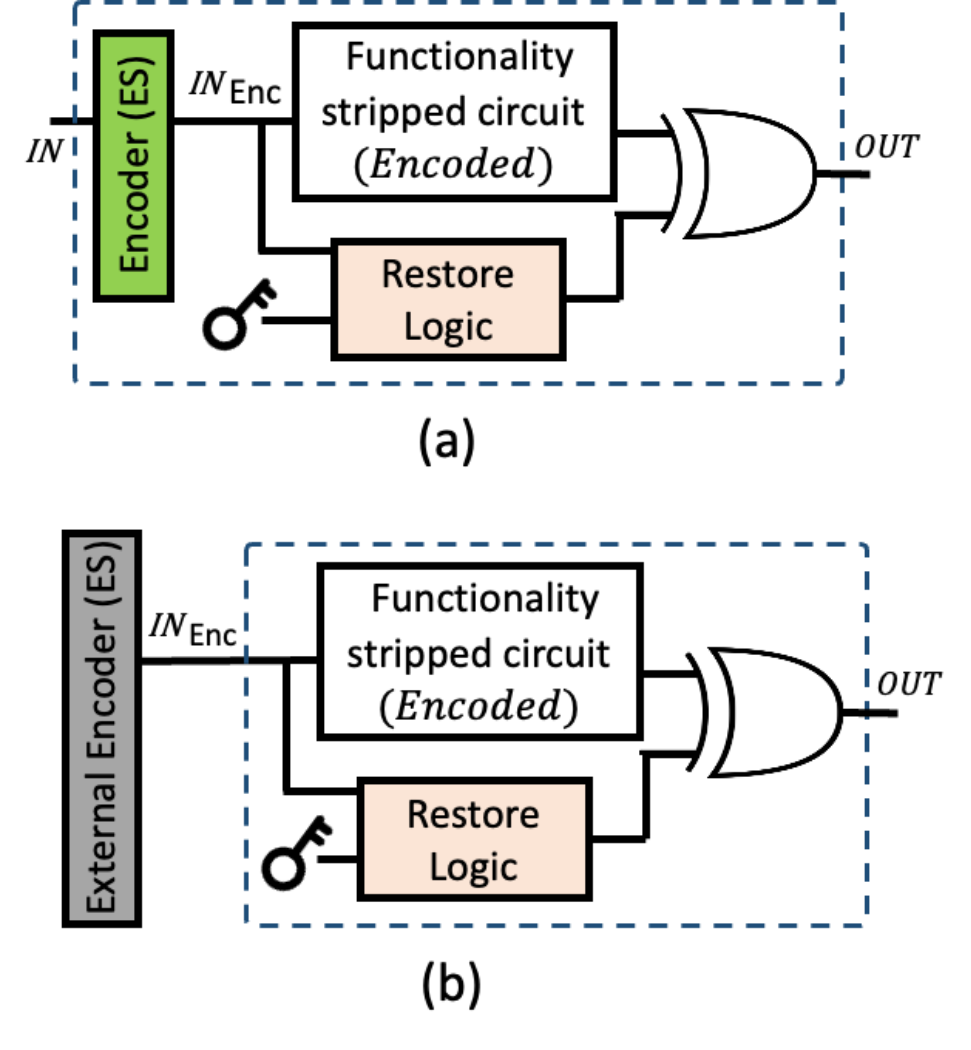}
  \caption{SCONE-protected designs: (a) with encoding realized in the hardware itself, (b) with external encoding.}
  \label{fig:scone_diag}
\end{figure}

\section{Entropy-Reduction Vulnerability of SCONE}
\label{subsec:entropy_vuln}

Let the original circuit have $n$ primary inputs
\begin{equation}
\mathbf{x} \in \mathbb{F}_2^n .
\end{equation}
SCONE enlarges the apparent interface by introducing $m$ additional encoded inputs, yielding an encoded vector
\begin{equation}
\mathbf{z} \in \mathbb{F}_2^{n+m}.
\end{equation}
The security intuition behind SCONE is that this expansion increases brute-force complexity from $2^n$ to $2^{n+m}$. Our analysis shows that this interpretation is not generally valid when the additional inputs are deterministic functions of the original inputs.

Assume that the added encoded inputs are represented by
\begin{equation}
\mathbf{e} \in \mathbb{F}_2^m
\end{equation}
and are generated by a linear encoding rule
\begin{equation}
\mathbf{e} = A \mathbf{x},
\qquad
A \in \mathbb{F}_2^{m \times n}.
\label{eq:lin_encoding}
\end{equation}
Then every valid encoded input vector is of the form
\begin{equation}
\mathbf{z}
=
\begin{bmatrix}
\mathbf{x} \\
A\mathbf{x}
\end{bmatrix}.
\label{eq:encoded_vector}
\end{equation}
Accordingly, the set of valid encoded inputs is
\begin{equation}
\mathcal{C}
=
\left\{
\begin{bmatrix}
\mathbf{x} \\
A\mathbf{x}
\end{bmatrix}
:
\mathbf{x} \in \mathbb{F}_2^n
\right\}
\subseteq \mathbb{F}_2^{n+m}.
\label{eq:valid_code_space}
\end{equation}

\begin{proposition}
\label{prop:dim_c}
The valid encoded-input set $\mathcal{C}$ is an $n$-dimensional subspace of the ambient $(n+m)$-dimensional space $\mathbb{F}_2^{n+m}$.
\end{proposition}

\begin{proof}
Consider the linear map
\begin{equation}
E : \mathbb{F}_2^n \rightarrow \mathbb{F}_2^{n+m},
\qquad
E(\mathbf{x})
=
\begin{bmatrix}
\mathbf{x} \\
A\mathbf{x}
\end{bmatrix}.
\end{equation}
The first $n$ coordinates of $E(\mathbf{x})$ are exactly $\mathbf{x}$. Hence, if
\begin{equation}
E(\mathbf{x}_1)=E(\mathbf{x}_2),
\end{equation}
then $\mathbf{x}_1=\mathbf{x}_2$, and $E$ is injective. Since $E$ is linear and injective on an $n$-dimensional domain, its image has dimension $n$. Because $\mathcal{C}=\operatorname{Im}(E)$, the result follows.
\end{proof}

Proposition~\ref{prop:dim_c} implies that the number of valid encoded vectors is
$|\mathcal{C}| = 2^n,$ not $2^{n+m}$. Therefore, the additional $m$ coordinates increase the nominal interface width, but they do not introduce $m$ additional independent degrees of freedom.

The same conclusion follows from information theory. Since $\mathbf{e}$ is a deterministic function of $\mathbf{x}$,
\begin{equation}
H(\mathbf{e}\mid \mathbf{x}) = 0,
\end{equation}
and therefore
\begin{equation}
H(\mathbf{x},\mathbf{e})
=
H(\mathbf{x}) + H(\mathbf{e}\mid\mathbf{x})
=
H(\mathbf{x}).
\label{eq:joint_entropy}
\end{equation}
where $H$ denotes Shannon entropy. Thus, the encoded interface does not carry more entropy than the original $n$-bit input space.

\subsection{White-Box Exact Recovery for With-ES Variant}
\label{subsec:whitebox_1a}

In this SCONE variant implemented in hardware, additional encoded inputs are explicitly realized inside the locked netlist. Here, the encoding matrix $A$ can be recovered exactly from the implemented logic. Each added input satisfies an equation of the form
\begin{equation}
e_i = \bigoplus_{j=1}^{n} A_{ij}x_j,
\qquad i \in \{1,\dots,m\},
\label{eq:extra_eq}
\end{equation}
which yields the binary matrix
\begin{equation}
A \in \mathbb{F}_2^{m \times n}.
\end{equation}
Applying Gaussian elimination over $\mathbb{F}_2$ to $A$ reveals the rank of the ES and identifies redundant derived-input equations.

The key observation is that even when all $m$ rows of $A$ are linearly independent, the effective dimension of the valid encoded-input set remains $n$ by Proposition~\ref{prop:dim_c}. Hence, the added inputs are linearly independent only as \emph{functions} of $\mathbf{x}$; they are not independent entropy-bearing interface variables. This yields a constructive white-box vulnerability: the advertised input-space expansion is algebraically reducible to the original $n$-dimensional domain.

\begin{figure*}[h]
    \centering
\includegraphics[height=6cm,keepaspectratio]{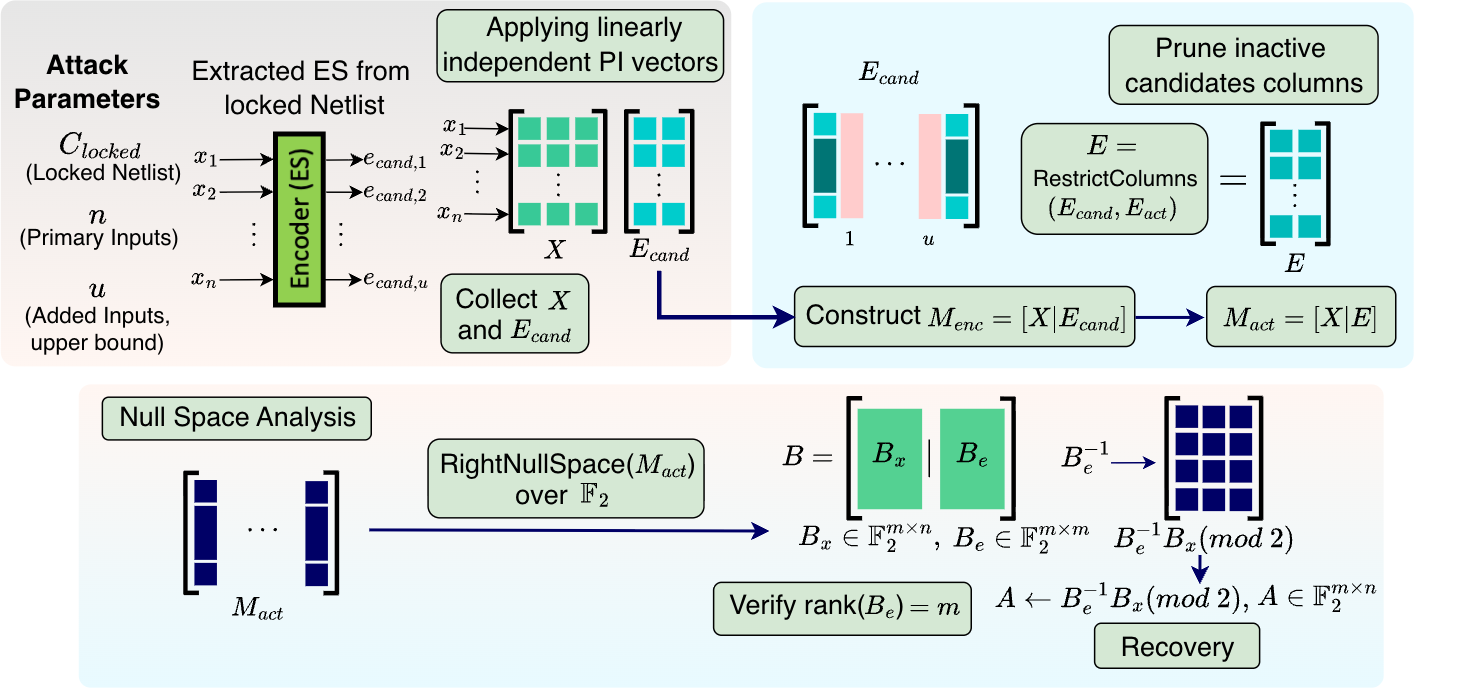}
    \caption{Illustration of the proposed white-box recovery flow for the with-ES SCONE setting. Starting from the locked netlist, the attacker probes candidate encoder outputs under linearly independent primary-input assignments, prunes inactive candidates, and algebraically recovers the implemented linear encoding relation.}
    \label{fig:attack1}
\end{figure*}

\begin{example}
Consider a toy SCONE-style instance with original primary inputs
\[
x=
\begin{bmatrix}
x_1\\
x_2\\
x_3
\end{bmatrix}
\in \mathbb{F}_2^3
\]
and added encoded inputs
\[
e=
\begin{bmatrix}
e_1\\
e_2
\end{bmatrix}
\in \mathbb{F}_2^2.
\]
Assume that the ES is implemented in hardware and the with-ES setting realizes
\[
e_1=x_1\oplus x_2,\qquad e_2=x_2\oplus x_3.
\]
In the white-box setting, consider the following three linearly independent probe assignments over the original PI space
\[
x^{(1)}=
\begin{bmatrix}
1\\
1\\
0
\end{bmatrix},
\qquad
x^{(2)}=
\begin{bmatrix}
1\\
0\\
0
\end{bmatrix},
\qquad
x^{(3)}=
\begin{bmatrix}
0\\
1\\
1
\end{bmatrix}.
\]
Stacking these probes as rows gives
\[
X=
\begin{bmatrix}
1 & 1 & 0\\
1 & 0 & 0\\
0 & 1 & 1
\end{bmatrix}.
\]

Evaluating the implemented ES on these probes yields
\[
e(x^{(1)})=
\begin{bmatrix}
0\\
1
\end{bmatrix},
\qquad
e(x^{(2)})=
\begin{bmatrix}
1\\
0
\end{bmatrix},
\qquad
e(x^{(3)})=
\begin{bmatrix}
1\\
0
\end{bmatrix}.
\]
Stacking the observed added-input values as rows gives
\[
E=
\begin{bmatrix}
0 & 1\\
1 & 0\\
1 & 0
\end{bmatrix}.
\]
The corresponding encoded sample matrix is
\[
M_{\mathrm{enc}}=
\left[
\begin{array}{c|c}
X & E
\end{array}
\right]
=
\begin{bmatrix}
1 & 1 & 0 & 0 & 1\\
1 & 0 & 0 & 1 & 0\\
0 & 1 & 1 & 1 & 0
\end{bmatrix}.
\]
Since $X$ is full-rank over $\mathbb{F}_2$, the valid encoded samples span a
$3$-dimensional subspace of the ambient space $\mathbb{F}_2^{3+2}$, i.e.,
\[
\operatorname{rank}(M_{\mathrm{enc}})=3.
\]

Moreover, the added-input relations can be recovered exactly from
\[
E = XA^T.
\]
Because
\[
X^{-1}=
\begin{bmatrix}
0 & 1 & 0\\
1 & 1 & 0\\
1 & 1 & 1
\end{bmatrix},
\]
we obtain
\[
A^T = X^{-1}E=
\begin{bmatrix}
1 & 0\\
1 & 1\\
0 & 1
\end{bmatrix},
\]
and therefore
\[
A=
\begin{bmatrix}
1 & 1 & 0\\
0 & 1 & 1
\end{bmatrix}.
\]
Hence,
\[
e=Ax,
\qquad
e_1=x_1\oplus x_2,
\qquad
e_2=x_2\oplus x_3.
\]

\end{example}
%
%
Figure~\ref{fig:attack1} summarizes the recovery flow from the
locked design to the active encoder structure. Once inactive
candidates are pruned, the remaining linear relations are sufficient
to recover the implemented SCONE encoding exactly.
Algorithm~\ref{alg:matrix_extraction} formalizes the white-box recovery flow for
the with-ES setting. The attack starts from the locked netlist
$C_{\mathrm{locked}}$, where the original PI count $n$ is known, whereas the
exact number of added inputs is not assumed a priori. Accordingly,
\textsc{IdentifyEncodedOutputCandidates} first extracts an upper-bounded set
$\mathcal{E}_{\mathrm{cand}}$ of wires that correspond to ES outputs.
This is a straightforward task for white-box structrual analysis of the netlist, where PIs feed into the ES and ES outputs feed into both the functionally stripped circuit and the restore logic
(Fig.~\ref{fig:scone_diag}).
Next, \textsc{ChooseProbeAssignment} selects PI vectors $\mathbf{x}$ such that
the sampled PI matrix $X$ remains full-rank as rows are accumulated. For each
accepted probe, \textsc{EvaluateCandidateESOutputs} evaluates the candidate ES
outputs and forms a valid encoded sample
$\mathbf{z}=[\,\mathbf{x}\parallel\mathbf{e}_{\mathrm{cand}}\,]$. Sampling
continues until $\operatorname{rank}(X)=n$, at which point the collected probes
span the original PI space.

Because only an upper bound on the added-input count is assumed initially, the
matrix $E_{\mathrm{cand}}$ may contain inactive columns. These are pruned by
retaining only candidate outputs that exhibit nonzero activity over the
collected probes, yielding the recovered active set
$\mathcal{E}_{\mathrm{act}}$. Its cardinality gives the recovered added-input
count,
\[
m = |\mathcal{E}_{\mathrm{act}}|.
\]
The reduced sample matrix is then written as
\[
M_{\mathrm{act}}=[\,X\mid E\,].
\]

Finally, \textsc{RightNullSpace} computes the right null space of
$M_{\mathrm{act}}$ over $\mathbb{F}_2$. Writing a row-basis of this null space
as $B=[\,B_x\mid B_e\,]$, the valid encoded samples satisfy
\[
B_x = B_e A,
\]
and therefore
\[
A = B_e^{-1}B_x \pmod 2.
\]
Hence, the algorithm recovers both the exact added-input count $m$ and the
linear relation of the extra inputs,
\[
\mathbf{e}=A\mathbf{x},
\]
thereby recovering the linear structure implemented by ES.

\subsubsection{Algorithm~1 Complexity}
Algorithm~1 is polynomial-time in the with-ES white-box setting. The
probe-collection loop does not enumerate the full $2^n$ original-input space;
it stops once the sampled PI matrix $X$ reaches rank $n$, i.e., once a
full-rank probe set has been collected. Thus, the attack requires only enough
linearly independent probes to span the original PI space.

Let $u$ denote the upper bound on the number of candidate encoded outputs and
let $m \le u$ be the recovered active encoded-input count. If
$T_{\mathrm{cand}}$ is the cost of identifying candidate encoded outputs and
$T_{\mathrm{eval}}$ is the cost of evaluating them on one accepted probe, then
the total runtime is
\[
T_{\mathrm{Alg1}}
=
T_{\mathrm{cand}}
+
O(n\,T_{\mathrm{eval}})
+
O\!\big(n (n+m)^2\big)
+
O(m^3 + m^2 n).
\]
Here, $O(n\,T_{\mathrm{eval}})$ accounts for evaluating the candidate ES outputs
on at most $n$ accepted probes, $O\!\big(n (n+m)^2\big)$ is the cost of Gaussian
elimination and right-null-space recovery on the reduced sample matrix, and
$O(m^3 + m^2 n)$ is the cost of inverting $B_e$ and reconstructing
$A = B_e^{-1}B_x$ over $\mathbb{F}_2$. Hence, in the with-ES setting,
recovering the implemented relation $\mathbf{e}=A\mathbf{x}$ is polynomial in $n$, $m$, and $u$.

\begin{algorithm}[t]
\caption{White-Box Recovery of the SCONE Encoding Matrix with Unknown Added-Input Count}
\label{alg:matrix_extraction}
\begin{algorithmic}[1]
\Require Locked netlist $C_{\mathrm{locked}}$ with ES implemented in hardware, known original PI count $n$, upper bound $u$ on the number of added encoded inputs
\Ensure Recovered active encoded-output set $\mathcal{E}_{\mathrm{act}}$ and encoding matrix $A \in \mathbb{F}_2^{m \times n}$ for some recovered $m \le u$

\Comment{\textbf{Phase 0: Structural identification of candidate encoded outputs}}
\State $\mathcal{E}_{\mathrm{cand}} \leftarrow \textsc{IdentifyEncodedOutputCandidates}(C_{\mathrm{locked}},u)$
\Comment{$|\mathcal{E}_{\mathrm{cand}}| = u$ candidate ES-output wires}
\State $X \leftarrow \emptyset$ \Comment{Row matrix in $\mathbb{F}_2^{0 \times n}$}
\State $E_{\mathrm{cand}} \leftarrow \emptyset$ \Comment{Row matrix in $\mathbb{F}_2^{0 \times u}$}
\State $M_{\mathrm{enc}} \leftarrow \emptyset$ \Comment{Row matrix in $\mathbb{F}_2^{0 \times (n+u)}$}

\Comment{\textbf{Phase 1: Collect $n$ linearly independent valid encoded samples}}
\While{$\operatorname{rank}(X) < n$}
    \State $\mathbf{x} \leftarrow \textsc{ChooseProbeAssignment}(n)$
    \If{$\mathbf{x}$ is linearly independent of the rows of $X$}
        \State $\mathbf{e}_{\mathrm{cand}} \leftarrow \textsc{EvaluateCandidateESOutputs}(C_{\mathrm{locked}},\mathcal{E}_{\mathrm{cand}},\mathbf{x})$
        \State $\mathbf{z} \leftarrow [\,\mathbf{x} \parallel \mathbf{e}_{\mathrm{cand}}\,]$
        \State Append row $\mathbf{x}$ to $X$
        \State Append row $\mathbf{e}_{\mathrm{cand}}$ to $E_{\mathrm{cand}}$
        \State Append row $\mathbf{z}$ to $M_{\mathrm{enc}}$
    \EndIf
\EndWhile

\Comment{\textbf{Phase 2: Prune inactive candidate outputs}}
\State $\mathcal{E}_{\mathrm{act}} \leftarrow \emptyset$
\For{$j = 1$ to $u$}
    \If{column $j$ of $E_{\mathrm{cand}}$ is not identically zero}
        \State $\mathcal{E}_{\mathrm{act}} \leftarrow \mathcal{E}_{\mathrm{act}} \cup \{\mathcal{E}_{\mathrm{cand}}[j]\}$
    \EndIf
\EndFor
\State $m \leftarrow |\mathcal{E}_{\mathrm{act}}|$
\State $E \leftarrow \textsc{RestrictColumns}(E_{\mathrm{cand}},\mathcal{E}_{\mathrm{act}})$
\State $M_{\mathrm{act}} \leftarrow [\,X \mid E\,]$

\Comment{\textbf{Phase 3: Recover null-space constraints on active coordinates}}
\State $\mathcal{N} \leftarrow \textsc{RightNullSpace}(M_{\mathrm{act}})$
\Comment{Compute the right null space over $\mathbb{F}_2$ by Gaussian elimination}
\State Form matrix $B \in \mathbb{F}_2^{m \times (n+m)}$ whose \emph{rows} are a basis of $\mathcal{N}$
\State Partition $B = [\,B_x \mid B_e\,]$, where $B_x \in \mathbb{F}_2^{m \times n}$ and $B_e \in \mathbb{F}_2^{m \times m}$

\Comment{\textbf{Phase 4: Reconstruct the encoding matrix}}
\State \textbf{assert} $\operatorname{rank}(B_e)=m$
\Comment{For active valid samples of the form $\mathbf{z}=[\,\mathbf{x}\parallel A\mathbf{x}\,]$, $B_e$ is nonsingular}
\State $A \leftarrow B_e^{-1} B_x \pmod 2$

\State \Return $\mathcal{E}_{\mathrm{act}}, A$
\end{algorithmic}
\end{algorithm}

\subsection{Black-Box Security Refutation for Without-ES Variant}
\label{sec:black-box-attack}
For the SCONE variant without the ES realized in hardware, the
locked circuit exposes only the encoded interface, and the explicit
relation in Eq.~\ref{eq:lin_encoding} is not present inside the netlist. In this setting, our
black-box analysis is restricted to the effective dimensionality of the
valid encoded-input space. We do not claim exact reconstruction of
the hidden encoding matrix from PI--PO access alone.

Still, our refutation is directly tied to SCONE's stated security rationale:
we ask whether the nominal expansion from $n$ to $n+m$ inputs
corresponds to $m$ additional independent brute-force dimensions.
Proposition~\ref{prop:dim_c} already answers this question negatively for the valid
encoded-input set. Specifically, valid operation is restricted to the
code space $\mathcal{C}$, whose dimension is $n$, not $n+m$.

To make this precise, let
\[
M_{\mathrm{enc}} \in \mathbb{F}_2^{k\times (n+m)}
\]
be a matrix whose rows are valid encoded vectors sampled from the
SCONE encoding. Since every row of $M_{\mathrm{enc}}$ lies in
$\mathcal{C}$, we have
\[
\operatorname{rank}(M_{\mathrm{enc}}) \le n.
\]
With sufficiently many independent samples drawn from the valid
encoded-input manifold,
\[
\operatorname{rank}(M_{\mathrm{enc}})=n.
\]
Hence, the dimension collapse is
\[
(n+m)-\operatorname{rank}(M_{\mathrm{enc}})=m.
\]
This is the relevant black-box conclusion. Even though the interface
visibly contains $n+m$ coordinates, the defender's own construction
restricts valid operation to an $n$-dimensional subset. Therefore, the
additional $m$ encoded coordinates cannot be counted as $m$
independent brute-force dimensions.

For an interpretation, consider that the null space of $M_{\mathrm{enc}}$ captures linear constraints that are
satisfied by all valid encoded vectors. Any non-zero null-space vector
therefore certifies that the encoded-input manifold is lower-dimensional
than the ambient $(n+m)$-dimensional interface. In the black-box
setting, we use this null-space structure as evidence of hidden linear
dependence and reduced effective dimensionality.

Algorithm~2 formalizes this black-box verification procedure for the
without-ES variant setting. It assumes access to valid encoded-input samples and
incrementally forms a row matrix
\[
M_{\mathrm{enc}} \in \mathbb{F}_2^{k\times W},
\]
where $W$ is the nominal encoded-interface width. Gaussian
elimination over $\mathbb{F}_2$ is then used to track the rank of
$M_{\mathrm{enc}}$ as additional valid samples are observed. Once the
rank saturates for $\tau$ consecutive samples, the recovered value $d$
is taken as the empirical dimension of the sampled valid encoded-input
space, and the quantity $W-d$ measures the number of redundant
interface dimensions.

Algorithm~2 is a conditional verification algorithm: given valid encoded-input observations, it determines whether they span the full nominal interface or only a lower-dimensional subset. Accordingly, its output is a dimensionality certificate for the valid encoded-input space, not a reconstruction of the hidden encoding rule.

\begin{algorithm}[t]
\caption{Black-Box Verification of Encoded-Input Dimension Collapse}
\label{alg:black_box_refutation}
\begin{algorithmic}[1]
\Require Black-box access to valid encoded-input samples, nominal interface width $W$, saturation threshold $\tau$
\Ensure Empirical dimension $d$ of the sampled valid encoded-input space and redundancy estimate $m_{\mathrm{red}} = W-d$

\State $M_{\mathrm{enc}} \leftarrow \emptyset$ \Comment{Row matrix in $\mathbb{F}_2^{0 \times W}$}
\State $d \leftarrow 0$
\State $t \leftarrow 0$ \Comment{Consecutive samples without rank increase}

\Comment{\textbf{Phase 1: Sample the valid encoded-input space}}
\While{$t < \tau$}
    \State $\mathbf{z} \leftarrow \textsc{SampleValidEncodedVector}()$
    \Comment{$\mathbf{z} \in \mathbb{F}_2^W$}
    \State $d_{\mathrm{new}} \leftarrow \operatorname{rank}\!\left(
    \begin{bmatrix}
    M_{\mathrm{enc}}\\
    \mathbf{z}^T
    \end{bmatrix}
    \right)$
    \Comment{Rank computed over $\mathbb{F}_2$ via Gaussian elimination}
    
    \If{$d_{\mathrm{new}} > d$}
        \State Append row $\mathbf{z}$ to $M_{\mathrm{enc}}$
        \State $d \leftarrow d_{\mathrm{new}}$
        \State $t \leftarrow 0$
    \Else
        \State $t \leftarrow t + 1$
    \EndIf
\EndWhile

\Comment{\textbf{Phase 2: Quantify dimension collapse}}
\State $m_{\mathrm{red}} \leftarrow W - d$
\Comment{Observed redundant interface dimensions}

\State \Return $d, m_{\mathrm{red}}$
\end{algorithmic}
\end{algorithm}

\subsubsection{Algorithm~2 Complexity}
Algorithm~2 is also polynomial-time in the number of collected samples and the
nominal encoded-interface width $W$. Let $k$ denote the number of sampled valid
encoded-input vectors observed before the rank saturates for $\tau$
consecutive samples. Each iteration appends one candidate sample and updates the
rank of a matrix in $\mathbb{F}_2^{k \times W}$ using Gaussian elimination, so
the dominant cost is rank maintenance over $W$ columns. In a dense
implementation, the overall runtime is
\[
T_{\mathrm{Alg2}} = O(kW^2),
\]
or equivalently $O(kW^2+\tau W^2)$ if the final saturation window is written
explicitly. The key point is that Algorithm~2 does not attempt any brute-force
enumeration over the nominal $2^W$ interface space; it only estimates the rank
of the sampled valid encoded-input set. Hence, the black-box verification cost
is polynomial in the interface width and the number of collected valid samples.

\subsection{Security Implications}
\label{subsec:security_implication}

The significance of the above analysis is not merely representational. SCONE's brute-force argument is predicated on the assumption that increasing interface width from $n$ to $n+m$ increases the number of independent attacker choices. Our results show that this assumption is invalid whenever the additional inputs are deterministic linear functions of the original inputs. In that case, the apparent interface expansion creates redundancy rather than entropy. Moreover, the assumption of a hidden or obfuscated encoding scheme (without-ES) as an additional security layer contradicts Kerckhoffs’s principle, which asserts that a system should remain secure even if all aspects of its design, except the key, are publicly known.

\section{Experimental Validations}
\label{sec:Exp}

\subsection{White-Box Recovery for With-ES variant}

\subsubsection{Setting}

We evaluate the proposed white-box recovery attack for the SCONE
implementation of with-ES variant on two benchmark families: ISCAS-85 and ITC-99.
For each benchmark, we generated SCONE-locked instances with a hardware-implemented with-ES variant~\cite{han2025scone,saha2025lockforge}.
The original PI count $n$ was obtained directly
from the benchmark interface, while the planted added-input count $m$ was swept across multiple benchmark-specific values to study recovery behavior under increasing encoded-interface expansion.

For each generated instance, we applied the white-box recovery
procedure of Algorithm~1 directly to the locked netlist and recorded the
planted added-input count $m$, the recovered count $\hat{m}$, the
nominal encoded-interface width $W=n+m$, the rank of the recovered
encoding matrix $\operatorname{rank}(A)$, and the total recovery
runtime. All experiments were executed using our Python-based prototype
on a standard desktop workstation. Since the present subsection focuses
only on the white-box attack, oracle interaction is not required here.
Likewise, stopping thresholds are not relevant in this setting, because
the attack operates by direct recovery of the implemented ES from the
locked netlist rather than by black-box rank saturation.

\begin{table}[tb]
\footnotesize
\centering
\caption{Summary of with-ES variant white-box recovery outcomes across the tested $m$-sweeps. Exact added-input recovery indicates $\hat{m}=m$ for all tested instances of the benchmark. Exact matrix recovery indicates exact reconstruction of the planted encoding matrix over the same sweep. Runtime trends are shown separately in the corresponding figure.}
\label{tab:fig1a_recovery}
\begin{tabular}{llcccc}
\toprule
Family & Circuits & $n$ & \makecell[c]{Tested \\ $m$ val.} & \makecell[c]{Added Inp. \\ Recov.?} & \makecell[c]{Matrix \\ Recov.?} \\
\midrule
\multirow{4}{*}{ISCAS-85}
& c432  & 36   & $\{4,8,16,20\}$       & \checkmark & \checkmark \\
& c6288 & 32   & $\{4,8,16,20\}$       & \checkmark & \checkmark \\
& c5315 & 178  & $\{10,20,40,80\}$     & \checkmark & \checkmark \\
& c7552 & 207  & $\{25,50,75,100\}$    & \checkmark & \checkmark \\
\midrule
\multirow{6}{*}{ITC-99}
& b15   & 485  & $\{50,100,150,200\}$  & \checkmark & \checkmark \\
& b17   & 726  & $\{100,200,300,400\}$ & \checkmark & \checkmark \\
& b18   & 3357 & $\{100,200,300,400\}$ & \checkmark & \checkmark \\
& b19   & 6666 & $\{100,200,300,400\}$ & \checkmark & \checkmark \\
& b21   & 522  & $\{100,200,300,400\}$ & \checkmark & \checkmark \\
& b22   & 767  & $\{100,200,300,400\}$ & \checkmark & \checkmark \\
\bottomrule
\end{tabular}
\end{table}

\subsubsection{Results}

Table~\ref{tab:fig1a_recovery} summarizes the white-box recovery
outcomes across the evaluated ISCAS-85 and ITC-99 benchmarks. For every
tested instance, the attack recovered the exact number of added encoded
inputs, i.e., $\hat{m}=m$, and reconstructed the planted encoding
matrix exactly. In all cases, the recovered matrix also satisfied
$\operatorname{rank}(A)=m$, confirming full-row-rank recovery of the
implemented linear encoding structure throughout the tested sweeps.

\begin{figure*}[tb]
    \centering
    \includegraphics[width=.94\textwidth, trim=0cm 0.0cm 0cm 0.0cm, clip]{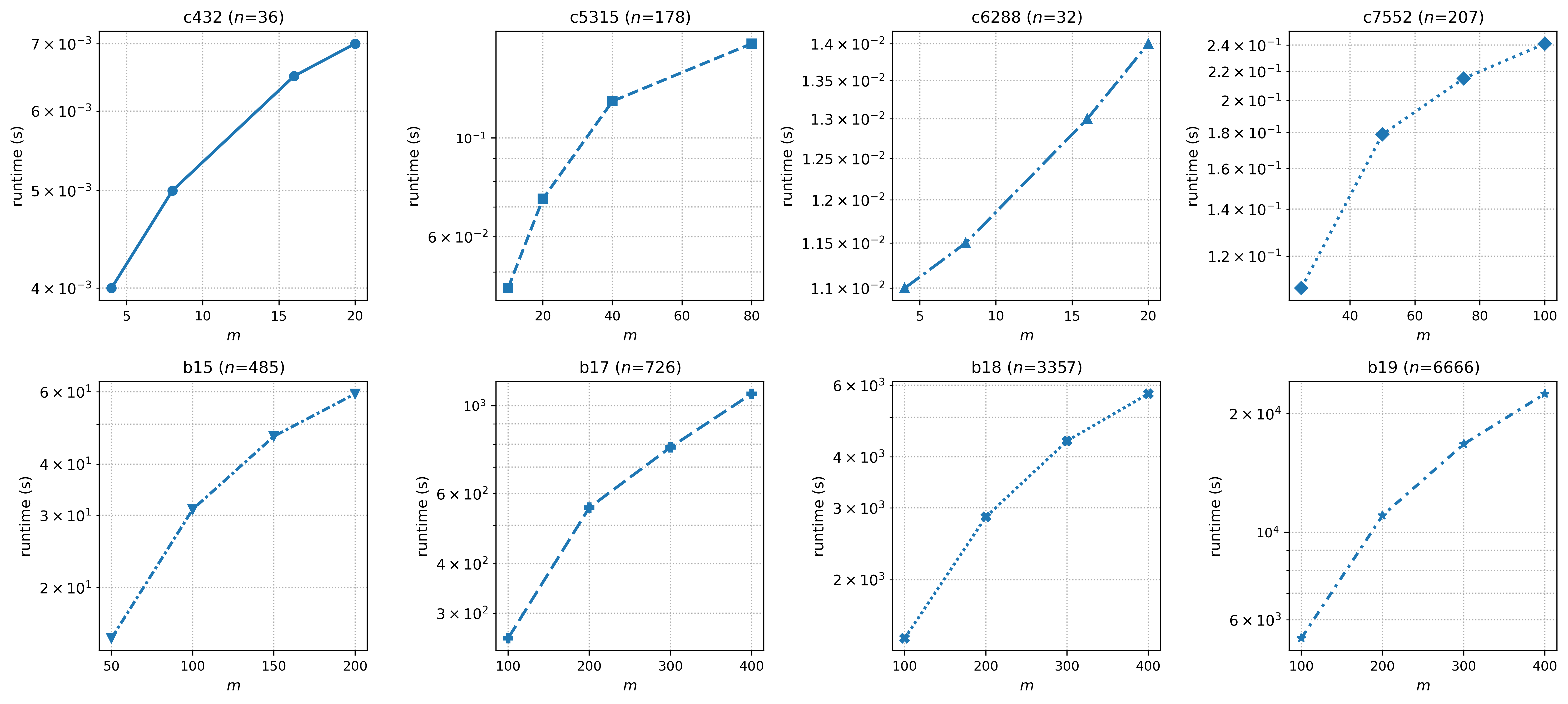}
    \caption{Per-benchmark runtime of the with-ES variant white-box recovery attack as a function of the planted added-input count $m$. Each subplot corresponds to one benchmark and shows the runtime trend over the tested sweep.}
    \label{fig:fig1a_runtime_small_multiples}
\end{figure*}

Figure~\ref{fig:fig1a_runtime_small_multiples} reports the recovery runtime for each benchmark as a function of the planted added-input count $m$. The results show the expected increase in runtime with growing benchmark size and encoded-interface width. The smaller ISCAS-85 circuits complete in milliseconds, while the larger ITC-99 circuits incur substantially higher recovery cost as both $n$ and $m$ increase. The largest evaluated instance is b19 at the largest tested sweep point, with runtime approximately $2.25\times 10^4$~s. Despite this increase, recovery remains exact over all evaluated instances.

Figure~\ref{fig:fig1a_sweep_coverage} shows the tested redundancy regimes in terms of the normalized ratio $m/n$. It indicates that exact recovery is observed not only across different absolute values of $m$, but also across a broad range of encoded-input expansions relative to the native PI dimension. Accordingly, the observed white-box vulnerability is not confined to a narrow operating region, but persists across substantially different benchmark scales and
redundancy regimes.

\begin{figure}[tb]
    \centering
\includegraphics[width=0.99\columnwidth,height=8cm,keepaspectratio]{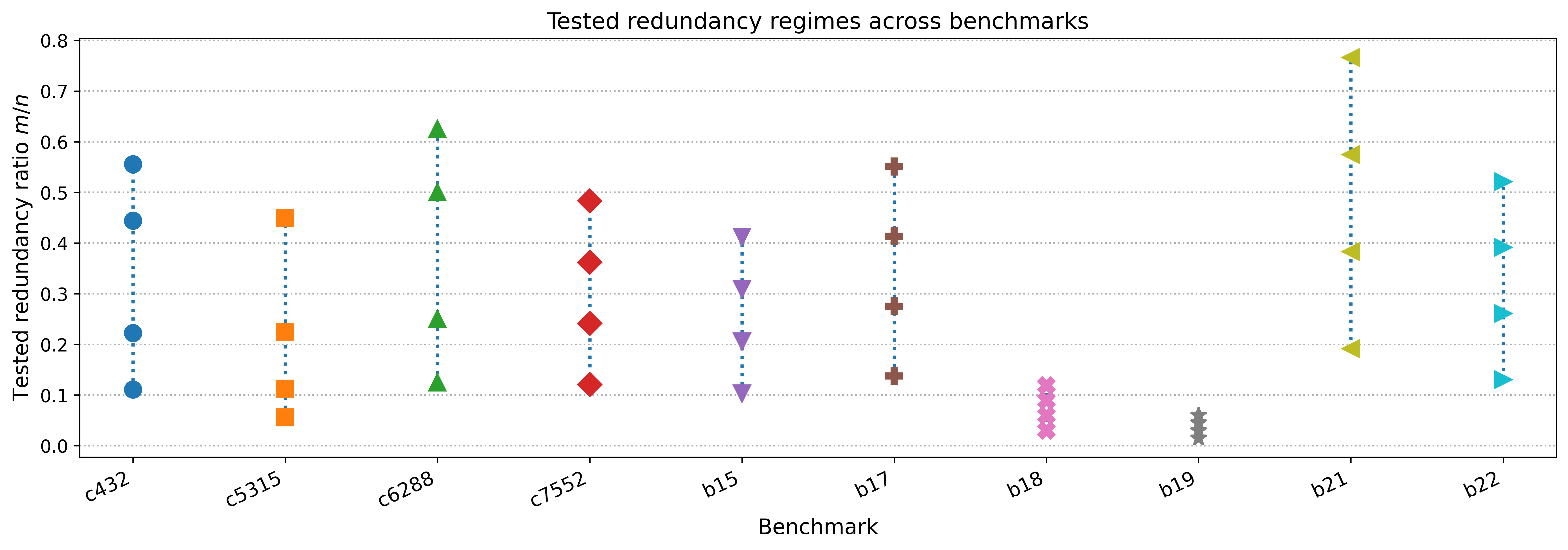}
    \caption{Tested redundancy regimes for with-ES variant, reported as the normalized ratio $m/n$ for each benchmark.}
    \Description{A benchmark-wise plot of normalized redundancy ratio m over n for the evaluated circuits. Each benchmark is shown on the horizontal axis, and markers indicate the tested m over n values for that circuit.}
    \label{fig:fig1a_sweep_coverage}
\end{figure}

\subsubsection{Case Study on ARM Cortex-M0}
We also conducted the white-box attack on an ARM Cortex-M0 processor~\cite{arm_cortex_m0}, a widely used processor in various microcontrollers~\cite{stm32f0_series, nxp_lpc1100}. The original RTL
consists of both sequential and combinational logic. To make the design suitable for LL, all flip-flop outputs are treated as primary inputs, and their inputs are treated as primary outputs.\footnote{%
To enable this conversion, logic synthesis is performed using Globalfoundry 55nm technology, without any optimization switches using the \texttt{exact\_map} option in the Synopsys Design Compiler \texttt{compile}
command. The resulting netlist is then modified to expose all flip-flop outputs as primary inputs and all flip-flop inputs as primary outputs, after which LL is applied to the design. Following LL best practices, the design can be transformed back to its original sequential form with the same set of inputs, outputs, and flip-flops, with portions of the combinational logic now logic-locked. One can perform the logic synthesis optimization on this netlist.}
This converted netlist 
contains 1020
inputs, 1746 outputs, and 15{,}308 logic assignments.

The attack achieved full recovery of the extra encoded inputs in both
evaluated settings: for the instance with $m=200$, it recovered all
200 added inputs in 109.63\,s, and for the instance with $m=400$, it
recovered all 400 added inputs in 226.42\,s.
These results
show that exact recovery remains feasible on a substantially larger processor-scale design.  As expected, runtime increases with the number of added encoded inputs, but the recovery remains exact in both cases.

\subsubsection{Summary}
Overall, the with-ES variant experiments establish the central white-box
vulnerability of the hardware-implemented ES. Across all evaluated
benchmarks, the attack recovers both the exact number of added encoded
inputs and the planted linear relation implemented by the encoder. As a
result, the nominal interface expansion from $n$ to $n+m$ does not
yield $m$ additional independent entropy-bearing input dimensions.
Instead, the added encoded coordinates remain deterministic functions of
the original PI vector and are directly recoverable from the locked
netlist.

\subsection{Black-Box Analysis for Without-ES Variant}

\subsubsection{Setting}
Section~\ref{sec:black-box-attack} shows that valid encoded inputs are confined to an $n$-dimensional subset of the ambient $(n+m)$-dimensional interface and therefore do not provide $m$ additional independent brute-force dimensions. The purpose of the following analysis is to evaluate the practical implications of this result in the black-box setting. Specifically, we study (i) the cost of obtaining valid encoded-input samples from the exposed $(n+m)$-input interface and (ii) whether, once such samples are available, Algorithm~2 saturates at the original PI dimension $n$ rather than the nominal width $W=n+m$.

Note that our black-box claim is narrower than the white-box one: we do not claim recovery of the hidden encoding from PI-PO access. Instead,
we ask whether operation over the visible interface behaves as an $(n+m)$-dimensional search space or remains confined to the $n$-dimensional manifold predicted by Proposition~\ref{prop:dim_c}.

To evaluate this, we generated without-ES locked benchmarks and performed 30 Monte Carlo trials per setting under uniform random sampling of the visible encoded interface. In each trial, the candidate input was drawn over full $W$-bit interface and labeled valid only for offline ground-truth evaluation using the planted generation rule. This rule is assumed to be unknown to the attacker. For each trial, we recorded the number of sampled inputs required to observe the first valid encoded-input sample and, when enough valid samples were obtained, the number required for Algorithm~2 to reach rank saturation. We also recorded the recovered sampled dimension $\hat d$ and the redundancy estimate $W-\hat d$. Here, ``enough valid samples'' means a sufficient number of linearly independent valid encoded input vectors for sampled rank to stabilize, as in Algorithm~2. Thus, the experiment measures the practical cost of exposing lower-dimensional valid manifold predicted by theory, not the recovery of the hidden encoding itself.

\subsubsection{Results}

\begin{figure}[tb]
    \centering
\includegraphics[width=.9\columnwidth,height=9cm,keepaspectratio]{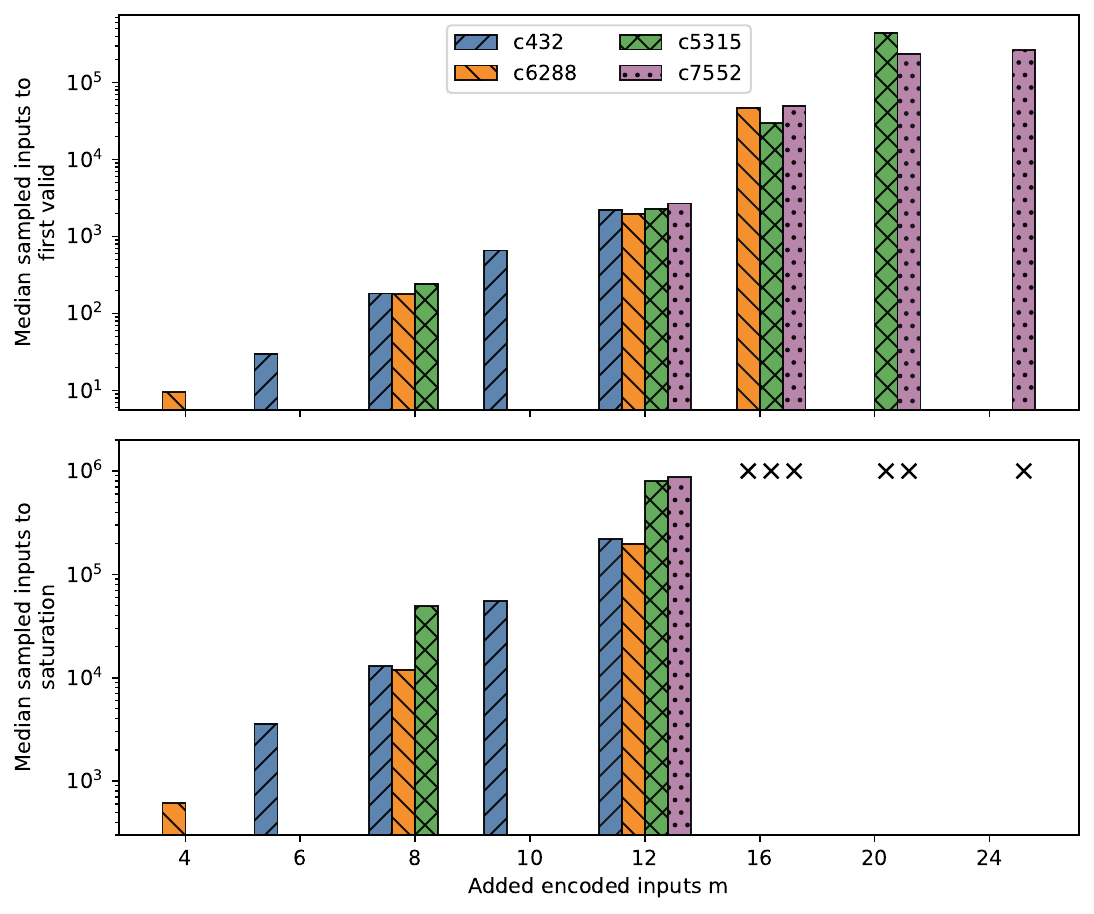}
    \caption{Black-box sampling cost for without-ES variant. As the added input count $m$ increases, the number of sampled inputs required to obtain a valid encoded-input sample rises sharply. The lower plot shows the corresponding cost of reaching rank saturation in successful settings. Here, $\times$ denote settings that did not reach saturation within the fixed sampling budget of $10^6$ sampled inputs.}
    \label{fig:blackbox-combined}
\end{figure}

Table~\ref{tab:blackbox-results} and
Figure~\ref{fig:blackbox-combined} summarizes the results. The upper
plot of Figure~\ref{fig:blackbox-combined} shows the median number of
sampled inputs required to obtain the first valid encoded input
sample. Across ISCAS-85 benchmarks this acquisition cost
increases sharply with $m$, indicating that valid encoded vectors
become progressively harder to encounter as the nominal interface
widens. The lower panel shows the median number of sampled inputs
required for Algorithm~2 to reach rank saturation in successful
settings. Bars marked by $\times$ denote settings that did not reach
saturation within the fixed sampling budget and therefore have no
reported saturation median; the corresponding table entries are shown
by a hyphen.

In all successful cases, the recovered dimension equals the
original PI count, i.e., $\hat d=n$, and the observed redundancy
satisfies $W-\hat d=m$. This is exactly the black-box consequence
predicted by Proposition~\ref{prop:dim_c} and used by Algorithm~2. Hence, whenever
enough valid encoded-input samples are obtained, the sampled
encoded-input space saturates at $n$, not at the nominal width
$W=n+m$.

The failed-saturation settings are also informative. For c6288 at $m=16$, c5315 at $m\in{16,20}$, and c7552 at $m\in{16,20,24}$, valid samples are observed, but saturation is not reached within sampling budget ($10^6$). Thus, in without-ES variant setting, the practical
bottleneck is valid-sample acquisition rather than rank computation.

\subsubsection{Summary}
Overall, these black-box experiments complement the theory in Section~\ref{sec:black-box-attack} by showing that the apparent expansion from $n$ to
$n+m$ inputs manifests in practice as increased sparsity of valid encoded vectors, not as $m$ additional independent brute-force dimensions.
That is, whenever saturation is reached, the sampled valid space saturates at dimension $n$, so the effective valid-space complexity remains governed by the original PI dimension rather than the nominal width $W=n+m$.

\begin{table}[tb]
\caption{Summary of black-box without-ES variant outcomes over 30 Monte Carlo
trials per setting. Here, \emph{sat. success} denotes the fraction of
trials reaching saturation within the budget, $\hat d$ is the
recovered sampled dimension for successful settings, and
$W-\hat d$ is the corresponding observed redundancy gap. Probe-cost
trends are shown separately in Figure~\ref{fig:blackbox-combined}.}
\label{tab:blackbox-results}
\centering
\begin{tabular}{llccc}
\toprule
Circuit & Setting & $W$ & sat. success & $\hat d$, $W-\hat d$ \\
\midrule
\multirow{4}{*}{c432}
  & $m=6$  & 42  & \checkmark & 36, 6  \\
  & $m=8$  & 44  & \checkmark & 36, 8  \\
  & $m=10$ & 46  & \checkmark & 36, 10 \\
  & $m=12$ & 48  & \checkmark & 36, 12 \\
\midrule
\multirow{4}{*}{c5315}
  & $m=8$  & 186 & \checkmark & 178, 8  \\
  & $m=12$ & 190 & \checkmark & 178, 12 \\
  & $m=16$ & 194 & -- & NA, NA \\
  & $m=20$ & 198 & -- & NA, NA \\
\midrule
\multirow{4}{*}{c6288}
  & $m=4$  & 36  & \checkmark & 32, 4  \\
  & $m=8$  & 40  & \checkmark & 32, 8  \\
  & $m=12$ & 44  & \checkmark & 32, 12 \\
  & $m=16$ & 48  & -- & NA, NA \\
\midrule
\multirow{4}{*}{c7552}
  & $m=12$ & 219 & \checkmark & 207, 12 \\
  & $m=16$ & 223 & -- & NA, NA \\
  & $m=20$ & 227 & -- & NA, NA \\
  & $m=24$ & 231 & -- & NA, NA \\
\bottomrule
\end{tabular}
\end{table}

\begin{table}[tb]
\caption{\small Evaluation of SCONE-locked circuits but with a non-linear encoding scheme, against representative structural and I/O-based attacks. None of the
listed attacks, including the vulnerability identified in this paper,
is successful on these instances.}
\label{tab:bfly_inputs}
\centering
\resizebox{0.9\linewidth}{!}
{
\begin{tabular}{c|c|c|c|c|c|c|c}
\hline
\textbf{Circuit} & \textbf{PIP} & \textbf{I/O SAT} & \textbf{SPS} & \textbf{ATR} & \textbf{FALL} & \textbf{SPI} & \textbf{Ours} \\ \hline

\multicolumn{8}{c}{\textbf{ISCAS’85}} \\ \hline
c432  & 38  & T.O. & \xmark & \xmark & \xmark & \xmark & \xmark \\
c1355 & 43  & T.O. & \xmark & \xmark & \xmark & \xmark & \xmark \\
c1908 & 35  & T.O. & \xmark & \xmark & \xmark & \xmark & \xmark \\
c7552 & 126 & T.O. & \xmark & \xmark & \xmark & \xmark & \xmark \\ \hline

\multicolumn{8}{c}{\textbf{ITC’99}} \\ \hline
b11 & 28  & T.O. & \xmark & \xmark & \xmark & \xmark & \xmark \\
b12 & 39  & T.O. & \xmark & \xmark & \xmark & \xmark & \xmark \\
b13 & 25  & T.O. & \xmark & \xmark & \xmark & \xmark & \xmark \\
b14 & 102 & T.O. & \xmark & \xmark & \xmark & \xmark & \xmark \\
b15 & 129 & T.O. & \xmark & \xmark & \xmark & \xmark & \xmark \\
b17 & 117 & T.O. & \xmark & \xmark & \xmark & \xmark & \xmark \\
b20 & 103 & T.O. & \xmark & \xmark & \xmark & \xmark & \xmark \\
b21 & 103 & T.O. & \xmark & \xmark & \xmark & \xmark & \xmark \\
b22 & 103 & T.O. & \xmark & \xmark & \xmark & \xmark & \xmark \\ \hline

\multicolumn{8}{c}{\textbf{Others}} \\ \hline
\multirow{2}{*}{\shortstack{ARM\\Cortex M0}}
& \multirow{2}{*}{1020}
& \multirow{2}{*}{T.O.}
& \multirow{2}{*}{\xmark}
& \multirow{2}{*}{\xmark}
& \multirow{2}{*}{\xmark}
& \multirow{2}{*}{\xmark}
& \multirow{2}{*}{\xmark} \\
& & & & & & & \\ \hline
\end{tabular}
}
\end{table}

\section{Potential Countermeasure}
\label{sec:Dis}
The core weakness of the current SCONE construction is that each added
encoded input is implemented as a deterministic linear function of the
original PI vector.
Introducing non-linear logic, such as AND/OR-based mixing or other higher-order Boolean dependencies, would break this linear
recoverability.

Such non-linear encoding would prevent the white-box attack on the encoder being represented as a
single matrix over $\mathbb{F}_2$. Consequently, the attacker would
face a substantially harder function-recovery problem than a
direct linear-structure extraction problem.
This difference is also reflected in Table~\ref{tab:bfly_inputs}, where a simple non-linear encoding\footnote{%
$in^{encoded}_i = in_{j_1} \oplus in_{j_2} \oplus \cdots \oplus in_{j_w} \oplus (in_{j_1} \land in_{j_2}),\; i \in [n+1, n+m] \cap \mathbb{N}$, where $j_1,\dots,j_w \in \{1,\dots,n\}$ are distinct indices selected uniformly at random and fixed thereafter.}
does not exhibit the vulnerabilities
identified in this paper and is \textit{actually} secure against the 
structural and I/O-based attacks which SCONE claims resilience against. While non-linearity alone is
not a complete proof, it is a principled first step toward
mitigating the vulnerability exposed.

Regarding the black-box vulnerability for the without-ES variant, our dimensionality-collapse result is specific to the linear encoding as well. Thus, replacing the ES with a non-linear one removes the grounds upon which the rank-based analysis relies.
  
\section{Conclusion}
\label{sec:Conclu}
This paper independently revisited the security claims of SCONE and showed that they
critically depend on how the encoding is realized. We analyzed both SCONE realizations, validated the theoretical
security analysis experimentally, and outlined a lightweight
non-linear mitigation that avoids the identified weakness.
Overall, this work highlight that nominal
interface expansion alone is not a sufficient basis for security in SCONE-like logic locking. For future work, we aim to develop formal design criteria for secure encoding constructions and evaluate their security-cost tradeoffs under broader attack models.


\bibliographystyle{ACM-Reference-Format}
\bibliography{bibliography}

\end{document}